\documentclass[sn-mathphys-num]{sn-jnl}

\usepackage{graphicx}%
\usepackage{multirow}%
\usepackage{amsmath,amssymb,amsfonts}%
\usepackage{amsthm}%
\usepackage{mathrsfs}%
\usepackage[title]{appendix}%
\usepackage{xcolor}%
\usepackage{textcomp}%
\usepackage{manyfoot}%
\usepackage{booktabs}%
\usepackage{algorithm}%
\usepackage{algorithmicx}%
\usepackage{algpseudocode}%
\usepackage{listings}%
\usepackage{braket}

\theoremstyle{thmstyleone}%
\newtheorem{theorem}{Theorem}
\theoremstyle{thmstyletwo}%

\theoremstyle{thmstylethree}%
\newtheorem{lemma}{Lemma}
\begin{document}

\title[Article Title]{A Symbol-Pair Decoder for CSS Codes}


\author[1]{\fnm{Vatsal Pramod} \sur{Jha}}\email{jha36@purdue.edu}

\author*[2]{\fnm{Udaya} \sur{Parampalli}}\email{udaya@unimelb.edu.au}

\author[3]{\fnm{Abhay Kumar} \sur{Singh}}\email{abhay@iitism.ac.in}

\affil*[1]{\orgdiv{Department of Computer Science}, \orgname{Purdue University}, \orgaddress{\street{Street}, \city{West Lafayette}, \postcode{47906}, \state{Indiana}, \country{United States}}}

\affil[2]{\orgdiv{School of Computing and Information Systems}, \orgname{The Univerdsity of Melbourne}, \orgaddress{\street{Street}, \city{Parkville}, \postcode{3010}, \state{Victoria}, \country{Australia}}}

\affil[3]{\orgdiv{Department of Mathematics and Computing}, \orgname{IIT(ISM) Dhanbad}, \orgaddress{\street{Street}, \city{Dhanbad}, \postcode{826004}, \state{Jharkhand}, \country{India}}}


\abstract{The relation between stabilizer codes and binary codes provided by Gottesman and Calderbank et al. is a celebrated result, as it allows the lifting of classical codes to quantum codes. An equivalent way to state this result is that the work allows us to lift decoders for classical codes over the Hamming metric to decoders for stabilizer quantum codes. A natural question to consider: Can we do something similar with decoders for classical codes considered over other metrics? i.e., Can we lift decoders for classical codes over other metrics to obtain decoders for stabilizer quantum codes? In our current work, we answer this question in the affirmative by considering classical codes over the symbol-pair metric. In particular, we present a relation between the symplectic weight and the symbol-pair weight and use it to improve the error correction capability of CSS-codes (a well-studied class of stabilizer codes) obtained from cyclic codes.}




\maketitle

\section{Introduction}\label{sec1}

Performing reliable quantum computation requires efficient quantum error correcting codes. The last two decades have seen significant theoretical advancements in this area \cite{css}, \cite{gottesman}, and \cite{GF(4)}. The stabilizer formalism introduced in \cite{gottesman} and \cite{GF(4)} provided the theoretical explanation for the existence of CSS codes (see\cite{steane}, \cite{css}), the Shor's 9-qubit code etc. Specifically, the aforementioned work converted the problem of obtaining stabilizer codes of length $n$ to that of obtaining self-orthogonal binary codes of length $2n$ with respect to the symplectic inner product and symplectic weight. We recall that the symplectic weight of a vector $(a|b):=((a_{0},...,a_{n-1})|(b_{0},...,b_{n-1}))\in(\mathbb{F}_{2}^{n})^{2}$ is defined as:
$$\mathrm{wt_{symp}}(a|b):=|\{1\leq i\leq n:(a_{i},b_{i})\neq (0,0)\}|.$$

In this paper, we show that the error correctability of CSS codes obtained from classical cyclic codes can be improved by lifting the syndrome decoder for codes over the symbol-pair metric that was mentioned in \cite{syndrome}.

The symbol-pair metric was introduced in \cite{sp} to abstract the error correction in high-density storage systems with the  symbol-pair weight of a vector $x:=(x_{0},x_{1},...,x_{n-1}) \in \mathbb{F}_{2}^{n}$ being defined as:
$$\mathrm{{wt_{sp}}}(x):=|\{i:(x_{i},x_{i+1})\neq (0,0)\}|,$$
where $i+1$ is calculated modulo $n$. 

An important reason to consider the symbol-pair metric is that it provides a natural way to treat a length $n$ vector as a length $2n$ vector and also bears an inherent relationship with the symplectic weight which we mention below. 

In our previous work \cite{VatsalISIT}, we establish the following relation between the symplectic weight and symbol-pair weight and use it to construct stabilizer codes from binary codes. For a vector $c:=(c_{0},...,c_{n-1})\in\mathbb{F}_{2}^{n}$ we have:
\begin{equation}\label{eq1}
\mathrm{wt_{sp}}(c)=\mathrm{wt_{symp}}(c|L(c)),
\end{equation}
where $L(c):=(c_{1},c_{2},...,c_{n-1},c_{0})$ is the cyclic left-shift operator on $\mathbb{F}_{2}^{n}$.

The above relation clearly serves as a strong motivation to consider codes over the symbol-pair metric and see its relation to stabilizer codes. Furthermore, the fact that for a vector $x\in\mathbb{F}_{2}^{n}\setminus\{(0,...,0),(1,...,1)\}$:
$$\mathrm{wt_{H}}(x)+1\leq \mathrm{wt_{sp}}(x)\leq 2\mathrm{wt_{H}}(x),$$
where $\mathrm{wt_{H}}(x):=\{i:x_{i}\neq 0\}$ is the Hamming weight of a vector $x:=(x_{0},...,x_{n-1})$, raises the possibility of improving the error-correction capability of stabilizer codes.

We would like to mention that in \cite{grassl} the space $(\mathbb{F}_{q}^{n})^{2}$ has been identified with $(\mathbb{F}_{q}^{2})^{n}$, which is similar to the symbol-pair abstraction but the paper \cite{grassl} does not use any result specific to the symbol-pair metric nor does it make any related claims in this regard.

\vspace{0.2cm}
We first present preliminaries on Quantum error correction and a relation between the Symbol-Pair metric and the Hamming metric. Then using this relation, we present syndrome decoding of codes over symbol pair metric. Using these results, we propose a new decoder for CSS codes by using syndrome decoding for codes over the symbol-pair metric. We claim that our decoder improves the error correctability of CSS codes using the syndrome decoder for codes over the symbol-pair metric.

The paper is organized as follows. In section II, we present the necessary background for Quantum error correction codes, Symbol pair metric and syndrome decoding of Symbol pair codes. In Section III, we present our novel Symbol-pair metric Quantum decoding algorithm. In Section IV we present our conclusions.

\section{Background}
\subsection{Quantum Error Correction}
It is known that a qu(antum)bit is modelled as a two dimensional complex vector space $\mathbb{C}^{2}$ with $\Ket{0},\ket{1}$ as an orthonormal basis i.e. every state $\ket{\psi}$ of a qubit can be written as $\ket{\psi}=\lambda_{1}\ket{0}+\lambda_{2}\ket{1}$ where $\lambda_{1},\; \lambda_{2}\in \mathbb{C}$ s.t. $|\lambda_{1}|^{2}+|\lambda_{2}|^{2}=1$. 
 The state space of $n-$qubits is represented by the space $(\mathbb{C}^{2})^{\otimes n}$ and has $\{\ket{x}:x\in\{0,1\}^{n}\}$ as its orthonormal basis.

The Pauli Group over $n$ qubits $P_{n}$ is the set:
$$P_{n}:=\{i^{\lambda} w_{1}\otimes w_{2}\otimes...\otimes 
w_{n}| \lambda\in\{0,1,2,3\}\; w_{j}\in \{I,X,Y,Z\}\},$$
equipped with the product operation of operators. Here $I,\; X,\; Y,\; Z$ represent the familiar Pauli gates.

The problem of finding stabilizer codes is essentially finding Abelian subgroups of $P_{n}$ (see \cite{GF(4)}) for more details. To elaborate, the code stabilized by an Abelian subgroup $S$ of $P_{n}$ defined as:
$$C(S):=\{\ket{\psi}:g\ket{\psi}=\ket{\psi}\; \forall g\in S\}.$$

The above characterization does not provide an efficient way to construct or search for stabilizer codes with good parameters. Now, to circumvent the aforementioned problem it was shown in \cite{GF(4)} that the problem of finding Abelian subgroups in $P_{n}$ is equivalent to the problem of finding self-orthogonal codes in $(\mathbb{F}_{2}^{n})^{2}$ where the orthogonality is defined with respect to the symplectic inner product i.e. $$<(a|b),(c|d)>_{symp}=<a,d>_{Euc}-<b,c>_{Euc},$$ where $<>_{Euc}$ is the Euclidean inner product over $\mathbb{F}_{2}^{n}$. 

The entire idea of finding stabilizer codes using codes over the symplectic inner product is captured by the following theorem.

\vspace{0.2cm}
\begin{theorem}
Let $S$ be a $n-k$ dimensional self-orthogonal subspace of $(\mathbb{F}_{2}^{n})^{2}$ considered with respect to the symplectic inner product. If  $$\underset{x\in S^{\perp_{symp}}\setminus S}{\min} \mathrm{wt_{symp}}(x)\geq d,$$ where $S^{\perp_{symp}}$ is the symplectic dual of $S$, then there exists a stabilizer code with parameters $[[n,k,\geq d]]$.
\end{theorem}
\begin{proof}
Refer to \cite{GF(4)}.
\end{proof}

Now, an easier way of characterizing a $[[n,k]]$ stabilizer code by a parity check matrix of the form:
$$\begin{bmatrix}
 H_{X} |H_{Z}\\
 \end{bmatrix},$$
 where $H_{X}$ and $H_{Z}$ are binary matrices of dimension $(n-k)\times n$ s.t. $H_{X}H_{Z}^{T}=0$.
In this paper, we focus on a specific kind of stabilizer codes known as CSS codes (see~\cite{steane},\cite{css}). The parity check matrix for a CSS code derived from codes $C_{1},C_{2}$ with $C_{2}\subseteq C_{1}$ appears in the form:
 $$\begin{bmatrix}
 H_{C_{1}} & 0\\
 0 & H_{C^{\perp}_{2}}\\
 \end{bmatrix}$$
 such that, $$H_{C_{1}}.H^{T}_{C^{\perp}_{2}}=0.$$

 The interesting thing about CSS codes is that their error-correctability can be captured by the underlying classical codes $C_{1},\; C_{2}$. Formally, we put this in the following Theorem:
 
\vspace{0.2cm}
 \begin{theorem}\label{css correction}
 Let $C_{1},C_{2}$ be binary linear codes. If $[n,k_{1},d_{H}(C_{1})]$ and $[n,n-k_{2},d_{H}(C_{2}^{\perp_{Euc}})]$ are the parameters of the linear codes $C_{1}$ and $C_{2}^{\perp_{Euc}}$ then the CSS code obtained from $C_{1}$ and $C_{2}$ has parameters
$[[n,n+k_{2}-k_{1},\geq min\{d_{H}(C_{1}),\; d_{H}(C_{2}^{\perp_{Euc}})\}]]$
\end{theorem}

In the following section, we review some definitions related to the symbol-pair metric.

\vspace{0.2cm}
\subsection{Symbol-Pair metric}
In high density data storage systems where the read head has a lower resolution than the write head, the read head is not able to distinguish between the adjacent symbols and hence reads adjacent symbols in every read operation. 

Mathematically, we can say that if the original codeword written by the high-resolution read head was $c:=(c_{0},...,c_{n-1})$ then the low-resolution read head will read it as $$\pi(c):=((c_{0},c_{1}),(c_{1},c_{2}),...,c_{n-1},c_{0})).$$
The vector $\pi(c)$ is referred to as the symbol-pair read vector of $c$. Now, it might happen that while performing the read operation some of the pairs of symbols might be read in error.

In order to correct against the symbol-pair errors the notion of symbol-pair weight was introduced in \cite{sp}, which for a vector $x:=(x_{0},x_{1},...,x_{n-1})\in \mathbb{F}_{2}^{n}$ is defined as:
$$\mathrm{wt_{sp}}(x):=|\{0\leq i\leq n-1:(x_{i},x_{i+1})\neq (0,0)\}|.$$
Now, consider the metrics $\mathrm{d_{H}}(x,y):=\mathrm{wt_{H}}(x-y)$,  $\mathrm{d_{sp}}(x,y):=\mathrm{wt_{sp}}(x-y),\; x,y\in\mathbb{F}_{2}^{n}$ defined by the two weight functions. For brevity, the functions $\mathrm{wt_{H}},\; \mathrm{wt_{sp}}$ will interchangeably be used to denote both the weight functions and the corresponding metrics as well. Moreover, the Hamming distance and symbol-pair distance of a linear code $C\subseteq\mathbb{F}_{2}^{n}$ is:
$$\mathrm{d_{H}}(C):=\min\{\mathrm{wt_{H}}(x):x\in C,\; x\neq0\}$$
$$\mathrm{d_{sp}}(C):=\min\{\mathrm{wt_{sp}}(x):x\in C,\; x\neq0\}.$$
The two metrics $\mathrm{wt_{H}}$ and $\mathrm{wt_{sp}}$ bear an interesting relation with one another and is stated in the following lemma:
\begin{lemma} (cf. \cite{sp})
 For binary vectors $x\in \mathbb{F}_\mathrm{2}^{n}$ with $0<\mathrm{wt}_\mathrm{H}(x)<n$ the following relations hold
 $$\mathrm{wt}_\mathrm{H}(x)+1\leq \mathrm{wt}_\mathrm{sp}(x)\leq 2\mathrm{wt}_\mathrm{H}(x)$$
and, 
$$\mathrm{d}_\mathrm{H}(x)+1\leq \mathrm{d}_\mathrm{sp}(x)\leq 2\mathrm{d}_\mathrm{H}(x).$$
 \end{lemma}

The reason to consider $\mathrm{wt_{sp}}$ rather than $\mathrm{wt_{H}}$ is essentially captured by the above inequality.

There have been several interesting results in terms of constructions and bounds pertaining to the symbol-pair metric like \cite{yaakobi}, \cite{elischo}, \cite{chee} etc. 
\newline
\subsubsection{Syndrome decoding for codes over the Symbol-Pair Metric}
In \cite{syndrome}, a syndrome decoding algorithm for codes over the symbol-pair metric was presented. The idea was based on ``identifying" a symbol-pair error using a pair of syndromes. We recall the syndrome decoding algorithm given in \cite{syndrome} in the current subsection.  
Let $$H:=\begin{bmatrix}
& &h_{1}& &\\
& &h_{2}& &\\
& &.& &\\
& &.& &\\
& &.& &\\
& &h_{n-k}& &
\end{bmatrix}
$$
be the parity-check matrix for a given classical code $C$ with parameters $[n,\; k,\; d_{p}]$. The symbol-pair parity check matrix for $H$, $\pi(H)$, is defined as:
$$\pi(H):=\begin{bmatrix}
& &\pi(h_{1})& &\\
& &\pi(h_{2})& &\\
& &.& &\\
& &.& &\\
& &.& &\\
& &\pi(h_{n-k})& &
\end{bmatrix}$$
The symbol-pair syndrome corresponding to a vector \\$u:=((\lhd u_{0},\rhd u_{0}),...,(\lhd u_{n-1},\rhd u_{n-1}))$ is defined as:
$$s^{p}:=u\pi(H)^{T},$$
where the inner product between vectors \\ $u:=((\lhd u_{0},\rhd u_{0}),(\lhd u_{2},\rhd u_{2}),...,(\lhd u_{n-1},\rhd u_{n-1}))$ and \\$v:=((\lhd v_{0},\rhd v_{0}),...,(\lhd v_{n-1}, \rhd v_{n-1}))$ defined as:
$$u.v:=((\lhd u_{0}\lhd v_{0},\rhd u_{0}\rhd v_{0}),...,(\lhd u_{n-1}\lhd v_{n-1},\rhd u_{n-1}\rhd v_{n-1})).$$

Another syndrome which can be defined for a vector \\$u:=((\lhd u_{0},\rhd u_{0}),...,(\lhd u_{n-1},\rhd u_{n-1}))$ is the neighbour-symbol syndrome defined as:
$$s^{n}:=(\lhd u_{0}+\rhd u_{n-1},\lhd u_{1}+\rhd u_{0},...,\lhd u_{n-1}+ \rhd u_{n-2}).$$

 Using the symbol-pair syndrome and nearest neighbour syndrome, the symbol-pair errors $e$ with weight $wt_{p}(e)\leq \lfloor\frac{d_{p}-1}{2}\rfloor$ can be corrected as given by the following theorem mentioned in \cite{syndrome}.
\vspace{0.2cm}
 \begin{theorem}\label{theorem sp correction}
For a binary code $C$ which can correct $\leq t_{p}$ errors, the pair of symbol-pair syndrome and neighbour pair syndrome $(s^{p},s^{n})$ is unique for each error vector $e$ with $wt_{p}(e)\leq t_{p}$.
 \end{theorem}
\vspace{0.2cm}
In the following section, we will improve the error-correctability for a particular class of CSS codes. 

\vspace{0.5cm}

\section{Quantum Decoding Algorithm }\label{sec2}

 \subsection{Decoding CSS codes using symbol-pair syndrome decoder}
 As already mentioned, the parity check matrix for CSS codes derived from codes $C_{1},C_{2}$ with $C_{2}\subseteq C_{1}$ looks like:
 $$\begin{bmatrix}
 H_{C_{1}} & 0\\
 0 & H_{C^{\perp}_{2}}\\
 \end{bmatrix}$$
 s.t. $$H_{C_{1}}.H^{T}_{C^{\perp}_{2}}=0$$

For simplicity, we will consider codes s.t. $C^{\perp}_{2}=C_{1}$ and $C_{1}^{\perp_{Euc}}\subseteq C_{1}$. 

We improve the error correctability of CSS codes in the following theorem.
\vspace{0.2cm}
 \begin{theorem}\label{improvised CSS correction}
For a binary cyclic code $C$ of length $n$ with $C^{\perp_{Euc}}\subseteq C$, the CSS code defined by $C,C^{\perp_{Euc}}$ can correct every error $e=(a|b)\in(\mathbb{F}_{2}^{n})^{2}$ with $wt_{symp}(a|b)\leq \lfloor\frac{d_{p}-1}{2}\rfloor$ and $wt_{H}(b+R(a))\leq \lfloor\frac{d_{H}-1}{2}\rfloor$, where $d_{H}:=d_{H}(C)$ and $d_{p}:=d_{p}(C)$. 
\end{theorem}
 \begin{proof}

 As already mentioned above, the parity check matrix for CSS code derived from $C,C^{\perp}$ will be:
$$\begin{bmatrix}
 H_{C} & 0\\
 0 & H_{C}\\
 \end{bmatrix},$$
 where $H_C$ is generated by the codeword $c\in C^{\perp_{Euc}}$ and its cyclic shifts. Let the rows of $H_{C}$ are denoted by $c_{1},...,c_{n-k}$, where $k=dim(C)$, dimensional of code $C$.

 Now, consider the following matrix:
 $$
 \begin{bmatrix}
 H_C& L(H_C)\\
 0 & H_{C}\\
 \end{bmatrix},$$
 where $L(H_C)$ is nothing but the matrix obtained by shifting rows of $H_C$ to the left. Interestingly, due to the fact that $C$ is cyclic, the above matrix also happens to be a parity check matrix for the CSS code generated by $C,C^{\perp_{Euc}}$. This is because the latter matrix can be obtained from the original parity-check matrix by a sequence of elementary row operations. As the null-space i.e. the vectors that get mapped to the zero vecctor by a linear map, remain unchanged under elementary row-operations, hence the new matrix also serves as a parity-check matrix for the CSS codes obtained from a cyclic code $C$.

 Now, consider an error $(a|b)$ with $wt_{symb}(a|b)\leq \lfloor\frac{d_{p}-1}{2}\rfloor$ and $wt_{H}(b+R(a))\leq d_{H}/2$.
 
Let the syndrome corresponding to $e_{flip}=(b|a)$ be denoted by $s:=(s_{1},...,s_{2(n-k)})$. This implies,
$$ \begin{bmatrix}
 H_C& L(H_C)\\
 0 & H_{C}\\
 \end{bmatrix}.e_{flip}^{T}=s^{T},$$
 Now consider the matrix multiplication corresponding to the submatrix:
 \begin{equation}\label{eq: 1}\begin{bmatrix}H_C& L(H_C)\end{bmatrix}(e_{flip})^{T}=(s^{'})^{T}.\end{equation}

 For every row $((c_{0},c_{1},...,c_{n-1}),(c_{1},c_{2},...,c_{0}))$ in the transformed parity-check matrix, we have: 
 \begin{align*}
 &<(c_{0},c_{1},...,c_{n-1},c_{1},c_{2},...,c_{0}),(b_{0},b_{1},...,b_{n-1},a_{0},a_{1},...,a_{n-1})>\\
 &=c_{0}(b_{0}+a_{n-1})+c_{1}(b_{1}+a_{0})+...+c_{n-1}(b_{n-1}+a_{n-2})
 \\
 &=c.(b+R(a))=s^{'}_{{i}}\\
 \end{align*}

 As $wt_{H}(b+R(a))\leq d_{H}/2$, hence the maximum likelihood decoding gives us $b+R(a)$.

 Now, consider the matrix multiplication corresponding to the submatrix:
 $$\begin{bmatrix}
 0 & H_{C}\\
 \end{bmatrix}.e_{flip}^{T}=s^{''}$$

 From this, we basically obtain:
 
 \begin{equation}\label{eq 2}H_{C}.a=s^{''}\end{equation}

As $<x,y>=<R(x),R(y)>$ for $x,y\in\{0,1\}^{n}$, hence we have 
$$R(H_{c}).R(a)=s^{''}.$$

Using the fact $C$ is a cyclic code and equations \ref{eq: 1} and \ref{eq 2} along with Gaussian elimination on the rows, we can obtain the symbol-pair syndrome and the neighbour-symbol syndrome corresponding to $(b|a)$ w.r.t the code $C$. Hence by Theorem \ref{theorem sp correction} we get our result.

\end{proof}

 The above result can easily be extended to the binary cyclic codes $C_{1},C_{2}$ with $C_{1}^{\perp}\subseteq C_{2}\subseteq C_{1}$. The error-correctability of the CSS codes then would be $\geq d_{p}(C_{1})/2$.

We will now show that our decoding scheme indeed substantially improves upon the known decoding scheme. We first show that every error $(x|z)$  with $wt_{symp}(x|z)\leq \lfloor\frac{d_{H}(C)-1}{2}\rfloor$ satisfies the conditions mentioned in Theorem \ref{theorem sp correction}.

To show this, we will need the following well-known theorem for quantum error- correction:
\vspace{0.2 cm}
\begin{theorem}\label{folklore}(folklore)
Consider a quantum error-correcting code $C$ on $n-$qubits. If C can correct bit-flip and phase-flip errors on $\leq t$ qubits, then it can correct arbitrary errors on $\leq t$ qubits, i.e. $C$ can correct errors $X(a)Z(b)$  where $wt_{symp}(a|b)\leq t$.
\end{theorem}
\vspace{0.2 cm}
Now, consider a bit flip error $(x|0)$ with $wt_{symp}(x|0)=wt_{H}(x)\leq d_{H}(C)/2$. As $d_{H}(C)\leq d_{p}(C)$ hence, $wt_{symp}(x|0)\leq d_{p}(C)/2$. Also, $wt_{H}(0+R(x))=wt_{H}(x)\leq d_{H}(C)/2$. Hence, every bit flip error $(x|0)$ on the first $n$-qubits with $wt_{H}(x)\leq d_{H}(C)/2$ satisfies the condition of Theorem \ref{theorem sp correction}. A similar observation holds for phase flip errors of the form $(0|z)$ with $wt_{H}(z)\leq d_{H}(C)/2$.   

Hence, by Theorem \ref{folklore}, every bit flip, and phase flip error acting on $\leq d_{H}(C)/2$ qubits can be corrected by our proposed decoding scheme. Thus, every error on $\leq d_{H}(C)/2$ qubits can be corrected by our proposed decoding scheme.

To show that our decoding scheme makes a substantial improvement, we invoke the following result concerning the symbol-pair distance of cyclic codes.
\vspace{0.2 cm}
\begin{theorem}\label{cyclic-sp:yaakobi}(\cite{yaakobi})
For every linear and cyclic code $C$, we have:
$$d_{p}(C)\geq \frac{3d_{H}(C)}{2}.$$
\end{theorem}

As an illustration for the claimed improvement in error correctability of CSS codes, consider the binary cyclic code $C$ generated by $g(x):=1+x+x^3\in\mathbb{F}_{2}[x]/<x^{7}-1>$. The generator polynomial for its Euclidean dual is $h(x):=1+x^2+x^3+x^4=g(x)(1+x)$. Hence $C^{\perp}\subseteq C$. The parameters of $C$ with respect to the Hamming metric is $[7,4,3]$ while with respect to the symbol-pair metric it is $[7,4,5]$. Now, consider the error $(x|y)=(1100000|0100000)$. What we can observe is that $wt_{symp}(x|y)=2\leq \frac{d_{p}(C)-1}{2}$ and $wt_{H}(R(x)+y)=1\leq d_{H}(C)/2$, hence by Theorem \ref{improvised CSS correction} the error $(x|y)$ can be corrected by our proposed decoding scheme but could not have been corrected by the previous decoding scheme for CSS codes as $wt_{symp}(x|y)=2>d_{H}(C)/2=1$.

\section{Conclusion}
In this paper we showed a possible connection between the symbol-pair metric and symplectic weight. This observation helped us in devising a new error correction scheme for CSS codes obtained from cyclic codes $C$ satisfying the dual-containing property i.e. $C^{\perp_{Euc}}\subseteq C$. The new error correction scheme proposed in Theorem \ref{improvised CSS correction} was derived from the syndrome decoding of codes over the symbol-pair metric while the improvisation was established by using Theorem \ref{cyclic-sp:yaakobi} from \cite{yaakobi}. What will be interesting to see is can we generalise Theorem \ref{improvised CSS correction} to more general errors or more general codes as compared to cyclic codes? Also, is it posssible to extend our decoding scheme for CSS codes to a bigger family of stabilizer codes?


\end{document}